\documentclass{article}

\usepackage[english]{babel}
\usepackage[final]{gpsmdms_2022}
\usepackage[dvipsnames]{xcolor}
\usepackage{tikz}
\usetikzlibrary{backgrounds}
\usetikzlibrary{arrows,shapes}
\usetikzlibrary{tikzmark}
\usetikzlibrary{calc}

\usepackage{amsmath}
\usepackage{amsthm}
\usepackage{amssymb}
\usepackage{mathtools, nccmath}
\usepackage{wrapfig}
\usepackage{comment}
\usepackage{graphicx} 
\usepackage{subcaption}

\usepackage{blindtext}

\usepackage{xspace}

\usepackage{array}
\usepackage{ragged2e}
\newcolumntype{P}[1]{>{\RaggedRight\hspace{0pt}}p{#1}}
\newcolumntype{X}[1]{>{\RaggedRight\hspace*{0pt}}p{#1}}

\usepackage{tcolorbox}

\usepackage{tikz}
\usetikzlibrary{arrows,shapes,positioning,shadows,trees,mindmap}
\usepackage[edges]{forest}
\usetikzlibrary{arrows.meta}
\colorlet{linecol}{black!75}
\usepackage{xkcdcolors} 

\usepackage{tikz}

\usepackage{amsmath}
\usepackage{graphicx}
\usepackage{amssymb}
\usepackage{amsthm}

\usepackage[dvipsnames]{xcolor}
\usepackage{natbib}
\setcitestyle{numbers,open={[},close={]}} 

\usepackage{mathtools, nccmath}
\usepackage{wrapfig}
\usepackage{comment}
\usepackage{algorithm}
\usepackage{algpseudocode}
\usepackage{blindtext}

\usepackage[colorlinks=true, allcolors=blue]{hyperref}
\newcommand{\comp}{\tau}
\newcommand{\s}{\{\theta_i,\comp_i\}}

\newtheorem{definition}{Definition} 
\newtheorem{proposition}{Proposition}
\newtheorem{theorem}{Theorem}

\usetikzlibrary{backgrounds}
\usetikzlibrary{arrows,shapes}
\usetikzlibrary{tikzmark}
\usetikzlibrary{calc}
\newcommand{\highlight}[2]{\colorbox{#1!17}{$\displaystyle #2$}}

\colorlet{mhpurple}{Plum!80}

\renewcommand{\highlight}[2]{\colorbox{#1!17}{#2}}

\author{%
  Pierre Thodoroff \\
  University of Cambridge	\\
  \texttt{pt440@cam.ac.uk} \\
  \And
  Markus Kaiser\\
  University of Cambridge	\\
  British Antarctic Survey\\
  \texttt{mk2092@cam.ac.uk} \\
    \And
  Rosie Williams\\
  British Antarctic Survey\\
  \texttt{chll1@bas.ac.uk	} \\
    \And
  Robert Arthern \\
  British Antarctic Survey\\
  \texttt{rart@bas.ac.uk}
    \And
  Scott Hosking \\
  British Antarctic Survey\\
  The Alan Turing Institute\\
  \texttt{jask@bas.ac.uk	}
  \And
  Neil Lawrence \\
  University of Cambridge	\\
  The Alan Turing Institute\\
  \texttt{ndl21@cam.ac.uk	}
  \And
  James Byrne \\
  British Antarctic Survey\\
  \texttt{jambyr@bas.ac.uk}
  \And
  Ieva Kazlauskaite \\
  University of Cambridge	\\
  British Antarctic Survey\\
  \texttt{ik394@cam.ac.uk	}
}

\DeclareMathOperator{\expect}{\mathbb{E}}
\newif\ifdraft
\draftfalse 

\title{Multi-fidelity experimental design for ice-sheet simulation}

\begin{document}
\maketitle
\section{Introduction}

Computer simulations are an essential tool in many scientific fields from molecular dynamics~\citep{hollingsworth2018molecular} to aeronautics~\citep{quagliarella2020open}. In glaciology, future predictions of sea level change require input from ice sheet models, which represent flow of the ice, mass changes due to accumulation of snow on the ice surface, and loss of ice through melting under floating ice shelves and iceberg calving. Ideally, these models need to be run over large geographical areas such as the Antarctic Ice Sheet and at sufficiently high resolution to accurately represent the dynamics of the ice flow ~\citep{pattyn2012results,pattyn2013grounding}. Moreover, due to uncertainties in the forcings and the parameter choices for such models, many different realisations of the model are needed to capture uncertainty in sea level contributions (SLC)~\citep{nias2019assessing}. For these reasons, producing robust probabilistic forecasts from an ensemble of model simulations over regions of interest can be extremely expensive for many ice sheet models and take several weeks. 



At the core of many numerical simulations lie differential equations approximated by numerical solvers that discretize the domain in time or space. While the accuracy of these solvers improves when step size is reduced, practitioners need to choose a size that is computationally feasible. The computational costs of running a simulator are intrinsically related to the accuracy of the model, leading to the question of finding the optimal trade-off between costs and accuracy. Multi-fidelity experimental design (MFED) is a strategy that models the high-fidelity output of a simulator by combining information from various resolutions in an attempt to minimize the computational costs of the process and maximize the accuracy of the posterior~\citep{stroh2022sequential,jakeman2022adaptive,gong2022multi}. Intuitively, for some problems, running a simulator with a low resolution might be sufficient to characterize its behaviour since the underlying dynamics might not involve high-frequency effects. 

In this paper, we present an application of MFED to an ice-sheet simulator~\citep{bradley2021wavi} and demonstrate potential computational savings by modelling the relationship between spatial resolutions. We analyze the regret of MFED strategies using theoretical results from sub-modular maximization and propose a new algorithm based on ideas from online learning (UCB) to efficiently optimize the exploration-exploitation trade-off along the computational axis.

\begin{figure}[htb]
    \centering 
\begin{subfigure}{0.3\textwidth}
  \includegraphics[width=\linewidth,height=2.9cm]{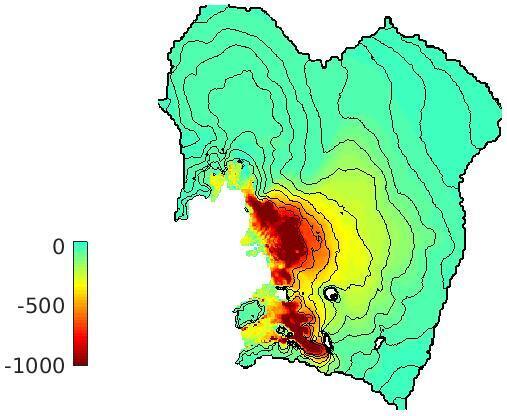}
  \label{fig:1}
\end{subfigure}\hfil 
\begin{subfigure}{0.3\textwidth}
  \includegraphics[width=\linewidth,height=2.9cm]{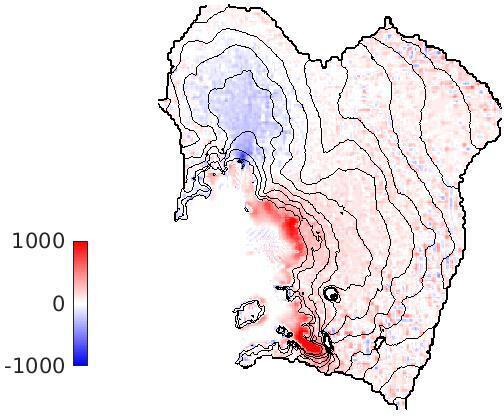}
  \label{fig:2}
\end{subfigure}\hfil 
\caption{Simulation results using the WAVI ice sheet model over the Amundsen Sea Sector of the West Antarctic Ice Sheet (details in Appendix \ref{app:plot}).  Left: modelled ice thickness change (in metres) over 100 years of a WAVI simulation at 3 km resolution. Right:  the difference in modelled ice thickness at $t=100$ years between a 10 km and a 3 km simulation.  }
\label{fig:ice}
\end{figure}



\section{Method}\label{sec:methods}
The goal of MFED is to model the behavior of a simulator  $f: \theta \times \comp \rightarrow \mathbb{R}^D$ in a cost-efficient manner such that $\sum_i c(\theta_i,\tau_i) \leq C$, where $\theta \in \mathbb{R}^N$ represents the parameters studied (e.g.\ temperature), $\comp \in  [0,\infty)^M$ the fidelity of the simulator ($0$ representing the highest fidelity), $c$ the non-negative cost-function of the simulator and $C$ the total computational budget. In the context of the ice-sheet simulator, the resolution $\tau$ represents the spatial resolution used (1-10km). The fidelity parameter balances the accuracy of the simulator with the computational costs (discretization factor). The Bayesian approach to experimental design is to infer a surrogate model of a simulator $f$ and then select points such as to minimize the model's uncertainty over the posterior distribution~\citep{chaloner1995bayesian}. There exist two important elements to define in the strategy: the model and the objective. 

\paragraph{Model} A commonly used model in MFED is to define a probability distribution over a space of functions distributed according to a Gaussian process such that $p(f) = GP(\mu(\cdot);k(\cdot,\cdot))$ where $\mu(\cdot)$ is the mean function and $k(\cdot,\cdot)$ is the covariance function. The computational axis $\comp$ can be modeled as a discrete set of fidelities~\citep{stroh2022sequential} or as a continuous one. In the discrete case, one strategy is to find a stationary relationship between the fidelities such as  $f(\theta,0)=\text{sin}(f(\theta,1)) + f(\theta,2)$. However, in the context of physical simulators, it seems unlikely that there exists such stationary relationships between fidelities. Furthermore, the number of models would explode due to the continuous nature of the fidelities. By modelling $f$ as a joint continuous space, the relationship between points of varying fidelities can be obtained through the covariance matrix. We consider kernels of the exponential quadratic families parameterized by a lengthscale $\beta$. 


\paragraph{Cost-adjusted utility} In experimental design, the aim is to reduce the uncertainty over our posterior distribution. In this work, we consider maximizing the reduction in conditional integrated variance (CIV)~\citep{gorodetsky2016mercer} such that the utility of adding a point $x_i = \{\theta_i,\tau_i\}$ is given by

\begin{equation}
	 U(x,X) =  \int \sigma_{X} (\theta,\comp=0) d\theta - \int \sigma_{x \cup X} (\theta,\comp=0) d\theta.
\end{equation}
where $\sigma_X$ is the variance of a GP fitted to the set of points $X$. The conditional on $\tau=0$ is due to the fact that we only care about the uncertainty over the highest fidelity. The aim of MFED is to reduce the uncertainty over the posterior while minimizing the computational budget used ($\sum_i c(x_i)$). Optimally reasoning about the budget allocation is an NP-hard combinatorial problem~\citep{martello1990knapsack} as illustrated by the knapsack problem which consists of selecting items with assigned weights and value such as to maximize the cumulative value under a total weight constraint. In special cases (discussed in section \ref{sec:theory}), a reasonable approximation can be obtained using a cost-adjusted strategy such that each point is selected to maximize the cost-utility ratio $\frac{U(x)}{c(x)}$. This strategy was proposed in~\citep{stroh2022sequential} and the related work are discussed in Appendix \ref{app:rel}. 



\paragraph{Uncertainty over objective}
One challenge is that CIV is dependent on the hyper-parameters of the GP and is unknown a priori. In particular, if we consider kernels from the exponential quadratic family, the lengthscale over the computational axis will have a large effect on the utility. Selecting a new observation that would maximize the CIV is thus challenging as the true lengthscale $\beta^*$ is unknown. To alleviate this issue, we learn a posterior distribution over $\beta$ at each time step and choose a specific $\beta$ based on an upper-confidence bound parameter $\nu$. In Section~\ref{sec:theory}, we discuss the role of $\nu$ in trading-off exploration and exploitation along the computational axis.

\begin{algorithm}[h]
\caption{MF-UCB}\label{alg:cap}
\begin{algorithmic}
\Require Confidence parameter $\nu$, prior over lengthscale $\widetilde{p}(\beta)$, budget C
\While{$C>0$}
    \State Select $\beta_i$ such that $\widetilde{p}(\beta \geq \beta_i) \leq \nu$
    \State $x_i \gets \text{argmax}_{x} \frac{U_{\beta_i}(x)}{c(x)}$
    \State Fit GP model with posterior distribution $\widetilde{p}(\beta)$ on $X \cup x_i$.
    \State $C \gets C-c(x_i)$
\EndWhile
\end{algorithmic}
\end{algorithm}


\section{Algorithm analysis}
\label{sec:theory}

In this section, we analyze the performance of the cost-adjusted strategy from Section \ref{sec:methods}. First, we discuss the regret incurred by using the cost-adjusted greedy approximation. Then we analyze the \emph{identification error} that captures the uncertainty existing over the optimal lengthscale. We showcase the existing underlying exploration-exploitation trade-off and its relationship with the confidence parameter $\nu$ introduced in Section~\ref{sec:methods}.

MFED can be analyzed as a set-maximization problem if the input space is discretized. Formally, the goal is to find a set of points $\{x_i\} \in X$ that maximizes a utility set function $F: 2^X \to \mathbb{R}$ while satisfying a cost constraint $\sum_{i} c(x_i) \leq C$. As a reminder, the algorithms described in Section~\ref{sec:methods} selects greedily the next point by maximizing the cost-adjusted variance reduction to form the set $X^G$. However, due to the uncertainty over the lengthscale, the algorithm cannot select the optimal sample $x_i^G$ with respect to the greedy strategy and we call this cost-normalized error $\epsilon$ (definition~\ref{def:add}).

%




\begin{theorem}[Theorem 6 \citep{streeter2008online}]\label{theorem:greedy}
   Let $X^G$ be the greedy set, $X^*$ the optimal set maximizing the monotone sub-modular function $F$, $\epsilon(x)$ the cost-normalized error (\ref{def:add}) and $c(x)$ the cost function. Then,
	\begin{equation}
		F(X^G) > \tikzmarknode{x}{\highlight{red}{$(1-e^{-1})$}} F(X^*) - \tikzmarknode{s}{\highlight{blue}{$\sum_{i=1}^L \epsilon_i(x^G_i) c(x^G_i)$}}
	\end{equation}
    \begin{tikzpicture}[overlay,remember picture,>=stealth,nodes={align=left,inner ysep=1pt},<-]
        \path (x.south) ++ (0,-0.5em) node[anchor=north east,color=red!67] (scalep){\textbf{greedy approximation error}};
        \draw [color=red!87](x.south) |- ([xshift=-0.3ex,color=red]scalep.south west);
        \path (s.south) ++ (0,-0.5em) node[anchor=north west,color=blue!67] (mean){\textbf{identification error}};
        \draw [color=blue!57](s.south) |- ([xshift=-0.3ex,color=blue]mean.south east);
    \end{tikzpicture}
\end{theorem}


\paragraph{Greedy approximation} The greedy error arises from the fact that finding the right combination of samples under a cost-constrained objective can be shown to be NP-hard (variant of the knapsack problem~\citep{karp1972reducibility}).  However, if we assume that the set-utility function $F$ is sub-modular, the greedy approximation which sequentially selects points that maximize the utility-cost ratio defined in Section~\ref{sec:methods} is guaranteed to be within a reasonable distance from the true answer as per Theorem~\ref{theorem:greedy}. Sub-modular functions (definition \ref{def:sub}) are set functions with the special property that the value of adding an element to the set decreases as the size of the input set increases (diminishing returns). The conditional integrated variance is an example of a sub-modular monotone function (proof in~\citep{krause2008robust}).

\paragraph{Identification error} The identification error represents the uncertainty over the lengthscale which prevents the algorithm from selecting the greedy set $X^G$. As the posterior distribution over the lengthscale converges, the step-wise identification error converges to 0 since the correct model has been identified. The rate of convergence $\epsilon$ could be derived from the rate of convergence of the hyper-parameters of a GP under some assumptions~\citep{li2020Bayesian}. However, the greedy strategy does not take into account the improvement over the lengthscale's uncertainty. To illustrate this, we study the following example where we can either sample a point $x_0$ with a cost of $2$ or two points  $x_1,x_2$ with a cost of $1$ (illustrated in Figure \ref{fig:ex}). The greedy approximation compares the utility of $F(x_0)$ and $F(x_1,x_2)$ and makes the implicit assumption that the error on both actions are the same $2\epsilon(x_0) = \epsilon(x_1) + \epsilon(x_2)$. However, the error $\epsilon$ changes after the first time step due to the GP being fitted on new data and $\epsilon_1(x_2)$ may be smaller than $\epsilon_0(x_2)$ where $\epsilon_i$ denotes the error at time step $i$. We call $\Delta_i(x) = \epsilon_{i}(x) - \epsilon_{i+1}(x)$ the exploration bonus. This implies that MFED methods are underweighting lower fidelities by not taking into account the improvement over the lengthscale. One way to alleviate this issue, is to encourage exploration by selecting an upper bound on the posterior distribution over $\beta$. Practically, this can be done by increasing the confidence parameter $\nu$. Intuitively, the values of $\nu$ and $\Delta$ should be correlated, where, reduction in identification error are encouraged using the exploration parameter. Selecting a confidence parameter close to $\nu=0.5$ leads to a greedy behavior over the high fidelities as opposed to high $\nu$ which leads to more exploratory behavior. There exists a balance between two objectives; one, is learning the relationships between the fidelities and the second one, is maximizing the utility.





\section{Experiment}

We consider the case of applying MFED to an ice sheet model~\citep{arthern2015flow}. The goal is to predict the contribution to the sea level from a specified region of Antarctica using the WAVI simulator. However, due to the unknown nature of some of the forcing parameters, we produce a probabilistic output over a varying set of input parameters. Specifically, we analyze the behavior of WAVI when varying the melting rate under the floating part of the ice sheet (the ice shelves). The melting rate is obtained using a parametrization of the basal melt rate described in~\citep{favier2019assessment,holland2008response} (see Appendix~\ref{app:melt}).

\paragraph{Setup \& analysis} A pre-defined ensemble of runs over a varying range of melt rate parameters and resolutions (5 and 10km) were obtained. We analyze the behavior MF-UCB would have had in this constructed dataset. In Figure~\ref{fig:wavi}, we show that depending on some input parameters of WAVI (whether melt is allowed within partially floating cells), the lower fidelities can be either informative or not.  When partial melting is considered, it becomes significantly more cost-efficient to query 10km resolutions, even though the most information could be gained by querying the high fidelity. In contrast, when partial melting is not considered, querying lower resolutions is not favourable, as the model has identified that the fidelities do not share information. CIV ensures that the lower fidelity is discarded in this case, which is reflected in the cost-utility surface of Figure~\ref{fig:wavi}. Concretely, the shape of the utility surface is modulated by the lengthscale learned (details in Appendix~\ref{app:training}) by the GP (high lengthscale when partial melt is enabled, low otherwise). We also display the difference in utility surface when considering the cost-adjusted utility. The cost of running the higher resolution (5km) is approximately ten times more expensive than the $10$km one which scales the utility in favour of the $10$km resolution with partial melt. The computational costs of fitting the GP model is negligible because a 5km simulation can take $1$ week to obtain. 

\begin{figure}[htb]
    \centering 
\begin{subfigure}{0.3\textwidth}
  \includegraphics[width=\linewidth]{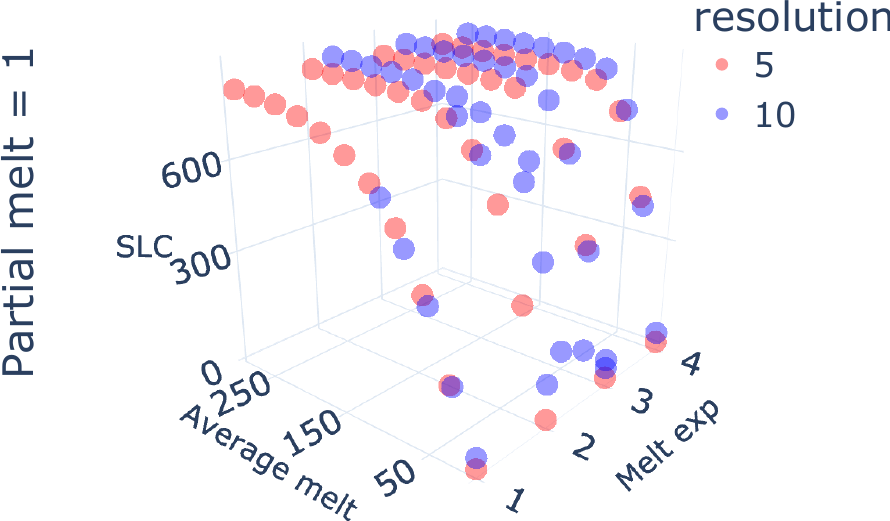}
  \label{fig:1}
\end{subfigure}\hfil 
\begin{subfigure}{0.3\textwidth}
  \includegraphics[width=\linewidth]{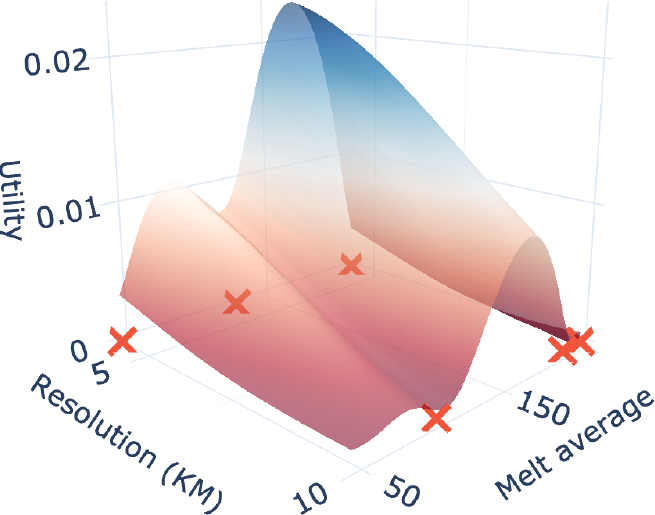}
  \label{fig:2}
\end{subfigure}\hfil 
\begin{subfigure}{0.3\textwidth}
  \includegraphics[width=\linewidth]{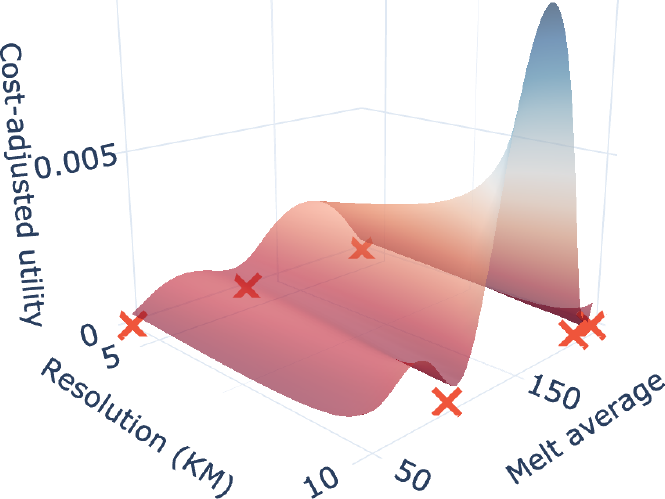}
  \label{fig:3}
\end{subfigure}
\begin{subfigure}{0.3\textwidth}
  \includegraphics[width=\linewidth]{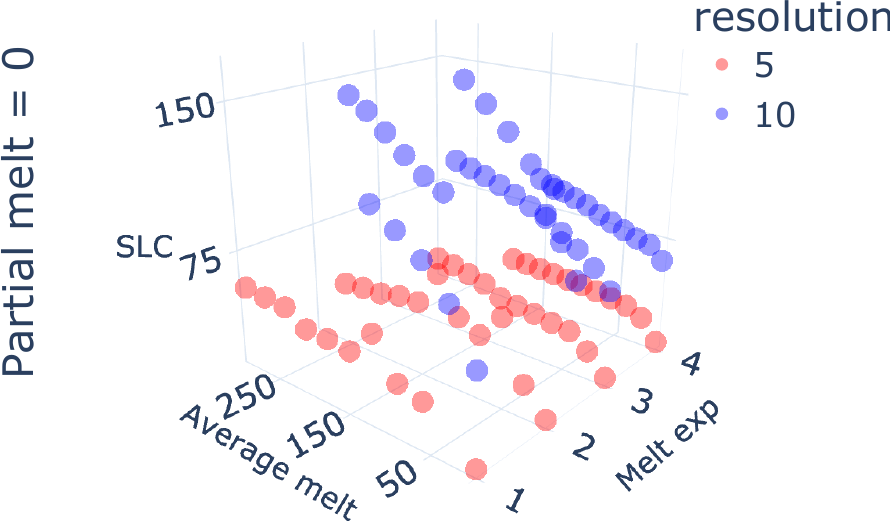}
  \label{fig:4}
\end{subfigure}\hfil 
\begin{subfigure}{0.3\textwidth}
  \includegraphics[width=\linewidth]{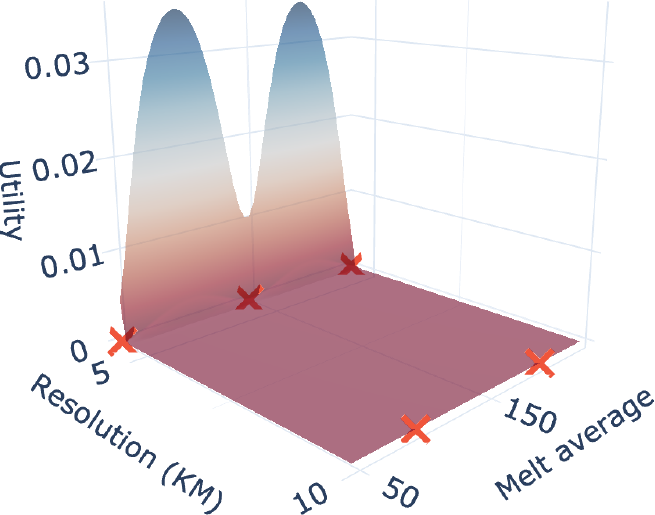}
  \label{fig:5}
\end{subfigure}\hfil 
\begin{subfigure}{0.3\textwidth}
  \includegraphics[width=\linewidth]{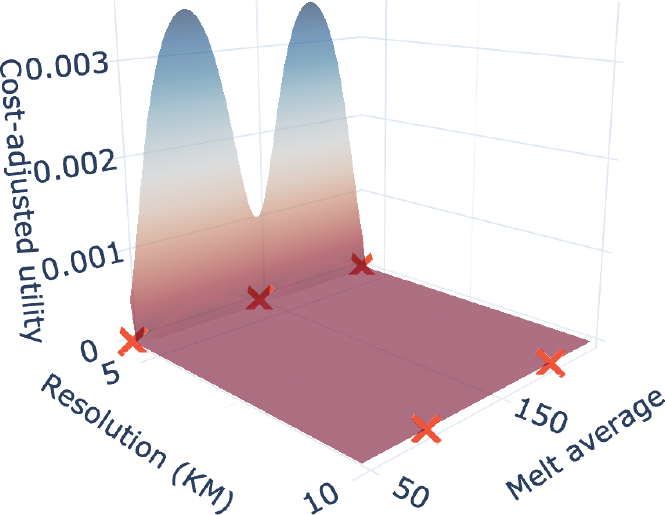}
  \label{fig:6}
\end{subfigure}
\caption{The left column shows the data distribution for varying average melt rate, melting exponent and resolution (5 and 10km). The middle column plots the variance reduction surface obtained when selecting a point after fitting a GP on 6 data points. The right columns displays the cost-adjusted utility. On the first row where partial melt is enabled, the lower resolution (10km) is informative and is a better cost-adjusted choice. On the contrary, when partial melt is disabled, the lower resolution is not informative and the algorithms learns to represent this aspect in the utility surface.}
\label{fig:wavi}
\end{figure}

\paragraph{Discussion} In this work, we investigate when the low-cost low-resolution simulations of an ice-sheet model can be used to predict the outcome of high-cost high-resolution runs. Our experiments summarized the information of the simulator into a single metric (SLC). However, this masks different ice-sheet behaviors that could prove informative. The next step of this project is to run the simulators based on the suggestions of MF-UCB on higher resolution (1-5km) and compare the performance of MFED using varying $\nu$ as well as other MFED methods. We hope to demonstrate in the future that MFED can be routinely incorporated in the glaciologists toolbox to obtain a cost-efficient probabilistic description of the change in sea level over a varying set of input parameters.  

\paragraph{Acknowledgements:} PT is funded by the Natural Sciences and Engineering Research Council and The Alan Turing Institute (AutoAI project). The compute to run this project has been generously provided by the British Antarctic Survey (Natural Environment Research Council) and Google Cloud (Climate Innovation Challenge). 


\bibliography{sample}
\clearpage
\section{Appendix}

\begin{figure}[h]
    \centering
    \includegraphics[scale=0.5]{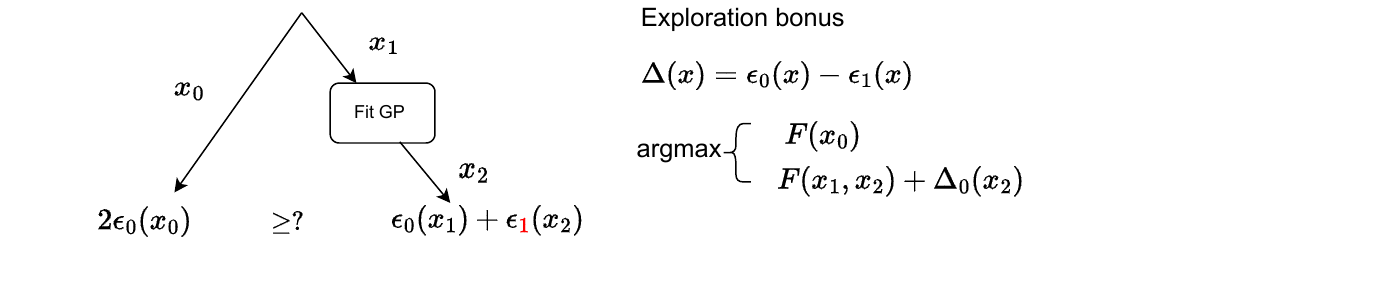}
    \caption{Example demonstrating the existing exploration-exploitation bonus.}
    \label{fig:ex}
\end{figure}

\subsection{Related work}\label{app:rel}

\paragraph{Experimental design:} There exists rich literature that studies the optimal placement of samples in experimental design~\citep{myers1972fundamentals,seltman2012experimental}. The two main strategies are to either use a space-filling design~\citep{joseph2016space} or a model-based strategy~\citep{geng2015gaussian}. With the space-filling strategy, the goal is to find a set of points that maximize the distance between points~\citep{viana2010algorithm}. In contrast, the model-based strategy casts a (probabilistic) model over the space studied and chooses points to minimize its uncertainty over the model~\citep{geng2015gaussian}. In model-based experimental design, the selection strategy is often divided into two phases. The first one focuses on finding points that are useful for learning the model. The second phase focuses on maximizing some defined utility functions such as the integrated variance.

\paragraph{Multi-fidelity experimental design:} Algorithm in multi-fidelity experimental can be divided into several categories based on whether they use a model and how the fidelities are analyzed. For example,~\citep{forrester2007multi} proposed a model-based two-stage process where several simulations are run at low fidelity and then use a sequential strategy to select high fidelity points. Another example is~\citep{xiong2013sequential} which uses hypercubes for sampling and doubles the fidelity until a satisfying criterion is reached. In this paper, we build our algorithm from the maximum rate of uncertainty reduction of~\citep{stroh2022sequential,villemonteix2009informational,bect2012sequential}. The main algorithmic difference lies in the fact that we select the lengthscale based on the upper bound $\nu$ of a learned distribution over the lengthscale. The justification for this modification comes from the fact that selecting the median lengthscale underestimates lower fidelities as explained in Section \ref{sec:theory}. Practically, this provides an exploratory phase into the algorithms that are often hard-coded by researchers. We also provide theoretical insights into the performance of the algorithm using sub-modular function. 

\paragraph{Multi-fidelity Bayesian optimization:} In multi-fidelity Bayesian optimization, the goal is to find the minima of a function by leveraging lower fidelities~\citep{song2019general,wu2020practical}. The concept of multi-fidelity is more popular in Bayesian optimization than in experimental design. The knapsack analysis for multi-fidelity while originating from experimental design has been used to understand multi-fidelity Bayesian optimization behavior. There have been several works using a cost-adjusted strategy in conjunction with continuous modelling over the computational axis~\citep{kandasamy2017multi,he2017optimization}. The main difference between BO and experimental design lies in the utility function maximized. In BO, the only thing that matter is the value of the output of the simulator. This is in contrast to experimental design which is agnostic to the output value of the simulator.  

\paragraph{Online learning:} Online learning~\citep{lattimore2020bandit} is a field that focuses on analyzing the behavior of decision-making strategies. The simplest setting considers comparing the cumulative reward of a discrete set of actions where the rewards are sampled from a Gaussian with an unknown mean. In this setting, several algorithms such as upper confidence bound~\citep{lai1985asymptotically} can be guaranteed to have an optimal regret. Optimality in this setting means that the regret grows at a logarithmic rate. The essential insight behind UCB is optimism in the face of uncertainty principle. Practically, this means that the algorithm will over-estimate arms with high uncertainty to encourage exploration. A similar strategy is used in our algorithm where the lengthscale is selected in an optimistic fashion (lower fidelity) until proven wrong. In future work we will investigate how similar optimality results can be obtained in our setting.

\paragraph{Sub-modular maximization:} The analysis performed in this field falls under the umbrella of combinatorial optimization. The main idea is to take variants of combinatorics problems such as knapsack and endow them with stronger assumptions (such as sub-modularity) to achieve reasonable approximations. The work developed by~\citep{krause2014submodular} is particularly relevant for two reason: one, they analyze a noisy setting which enables the incorporation of uncertainty over the lengthscale. Second, many of their applications are on experimental design which enables us to leverage some of their theoretical results concerning the integrated variance. Finally, several papers draw a similar connection to online learning~\citep{krause2008robust}. 

\subsection{Sub-modular function}

\begin{definition}[Sub-modular function]\label{def:sub}
    For every $X,Y \subseteq \Omega$ with $X \subseteq Y$ and $\forall x \in \Omega \setminus Y$, we have that 
    \begin{equation}
        F(X \cup x) - F(X) \geq F(Y \cup x) - F(Y)
    \end{equation}
\end{definition}
\begin{definition}[\cite{streeter2008online}]\label{def:add}
	For a cost function $c$ and a sub-modular function $U$ over a set $Z$, we define the additive error $\epsilon(x)$ as:
	\begin{equation}
	\epsilon(x) = \max_{y \in Z} \frac{F(X \cup y)}{c(y)} -  \frac{F(X \cup x)}{c(x)} 
\end{equation}
where $x,y \in Z$ are points and X is a set of points in $Z$.
\end{definition}

\paragraph{Caveat theorem 1}
The caveat in this bound is that it only applies to specific $T=\sum_i c(x_i)$ that we are given by applying the greedy strategy. Practically, the greedy strategy take one element at a time and add up the $c(x_i,\tau_i)$ giving us the T for which the theorem holds. In other words, we compare the performance of the greedy strategy after a total cost of $T$ with the optimal strategy given a budget of $T$. Extensions can be found in \citep{krause2008robust} to relax this assumption. 


\subsection{Ice-sheet simulator}\label{app:melt}
The simulator used in this paper is called WAVI (wavelet-based, adaptive-grid, vertically integrated ice sheet model)~\citep{arthern2015flow}. It is a vertically integrated, three dimensional ice sheet model which includes both the membrane stresses in the ice and the effects of vertical shear in order to simulate flow of both grounded and floating ice. A subgrid parameterisation is used to represent the movement of the grounding line (the boundary between the grounded and floating ice). Since ice retreat has shown to be sensitive to how melting is applied in cells containing the grounding line (e.g. Seroussi and Morlighem, 2018), a subgrid melting parameterisation is applied, where melt can be applied in proportion to the floating area of the cell (partial melt=1), or melt can only be prescribed in cells that are fully floating (partial melt=0). There is no calving in the WAVI model (the ice shelf front is fixed and cannot advance or retreat). For a complete introduction to the simulator refer to 
~\citep{arthern2015flow}. 

The model is initialised via data assimilation, inverting for basal drag and viscosity using accumulation, surface velocities and rates of change of surface elevation~\citep{arthern2015flow}. The data used to initialise the simulator for this experiment includes bathymetry from~\citep{fretwell2013bedmap2}, modified to account for data from~\citep{millan2017bathymetry}, DEM from Cryosat2~\citep{slater2018new}, updated from 2013 to 2017 using dh/dt, dh/dt from ICESat-2~\citep{smith2019land}, and surface velocities from annual MEaSUREs 16/17~\citep{mouginot2017comprehensive}. Other data sets are as in~\citep{arthern2015flow}. The initial state thus corresponds to the ice sheet in approximately 2017. For the simulations included in this experiment, we run the model only over the sector of the West Antarctic Ice Sheet that drains into the Amundsen Sea, as shown in Figure~\ref{fig:ice}. 


\subsubsection{Parameters of interest}
There exists many input parameters in the WAVI simulator that can be modified, such as parameters related to the prescription of basal melt rate, accumulation rate, and ice rheology. We restrict our analysis to a subset of parameters of major importance, namely those that control the basal melt forcing. In the absence of a coupled ice-ocean model, a basal melt parameterisation is used to prescribe the melt rate under the floating part of the ice sheet (the ice shelves). In these simulations, the basal melt is parameterised using an expression described in~\citep{favier2019assessment,holland2008response}, which is expanded to consider different values of the exponent, denoted $m$. The parameter $\gamma^T$ is calibrated to set the average basal melt at the start of the simulation, $t=0$. Simulations with different combinations of $m$, average melt and resolution are then run forward for $100$ years of simulation time, and the output that we aim to emulate is the total cumulative sea level contribution (SLC) after 100 years. Two sets of simulations are performed: one with no melt in partially floating cells (partial melt=0), and one with proportional melt in partially floating cells (partial melt=1), as described above.

\subsubsection{Details of simulation in Figure~\ref{fig:ice}}\label{app:plot}
In Figure~\ref{fig:ice}, the melt parameters are set as $m=2$, average melt at $t=0$ years is $54$ m/a, and there is no melt in partially grounded cells (partial melt=0). The thickness from the $10$ km simulation is interpolated onto the 3km grid, for comparison. In both plots, elevation contours for the 3km run at $t=100$ years are shown at $200$ metre intervals.

\subsection{Training GP}\label{app:training}
The kernel used is an RBF kernel with a lengthscale for each dimension of the data initialized at the value 1. We use Hamiltonian Monte Carlo (HMC) to learn a posterior distribution over the lengthscale parameter $\beta$. HMC is run for 500 steps and the last 250 are used to approximate the posterior distribution. For visualization purposes we train the lengthscale on only 6 datapoints which can result in somewhat unstable training. However, for the real application of MFED on ice sheet, the total number of runs will easily exceed 50. The codebase is developed using GPy~\citep{gpy2014} and Emukit~\citep{emukit2019}. We normalize the input and output of the simulator and add a small jitter to the resolution column to avoid numerical instabilities. We use an upper confidence bound $\nu$ of $0.9$.

\clearpage

\end{document}